\documentclass[11pt]{article} 
\usepackage[letterpaper,hmargin=1in,vmargin=1.25in]{geometry}
\usepackage{amsmath}                    
\usepackage{amssymb}                    
\usepackage{amsthm}                     
\usepackage{amsfonts}                   
\usepackage{cite}                       
\usepackage{hyperref}                   
\hypersetup{colorlinks=true,linkcolor=[rgb]{0.75,0,0},citecolor=[rgb]{0,0,0.75}}
\usepackage{enumitem}                   
\usepackage{graphicx}                   

\title{On the Combinatorial Complexity of Approximating Polytopes}

\author{%
	Sunil Arya\\
		Department of Computer Science and Engineering \\
		Hong Kong University of Science and Technology \\
		Clear Water Bay, Kowloon, Hong Kong\\
		arya@cse.ust.hk \\
		\and
	Guilherme D. da Fonseca\\
		Universit\'{e} d'Auvergne and LIMOS\\
		Clermont-Ferrand, France\\
		guilherme.dias\_da\_fonseca@udamail.fr \\
		\and
	David M. Mount\\
		Department of Computer Science and \\
		Institute for Advanced Computer Studies \\
		University of Maryland \\
		College Park, Maryland 20742 \\
		mount@cs.umd.edu \\
}

\newcommand{\floor}[1]{\left\lfloor #1\right\rfloor}

\newcommand{\ang}[1]{\left\langle #1 \right\rangle}
\newcommand{\RE}{\mathbb{R}}
\newcommand{\eps}{\varepsilon}

\newcommand{\RR}{\mathcal{R}}
\newcommand{\CC}{\mathcal{C}}
\newcommand{\HH}{\mathcal{H}}
\newcommand{\WW}{\mathcal{W}}

\newcommand{\diam}{\mathrm{diam}}
\newcommand{\radius}{\mathrm{radius}}
\newcommand{\conv}{\mathrm{conv}}

\newcommand{\vol}{\mathrm{vol}}
\newcommand{\width}{\mathrm{width}}

\newcommand{\ray}{\mathrm{ray}}
\newcommand{\etal}{\textit{et al.}}
\renewcommand{\tilde}[1]{\widetilde{#1}}
\newtheorem{theorem}{Theorem}[section]
\newtheorem{lemma}[theorem]{Lemma}

\begin{document}

\maketitle

\begin{abstract}
Approximating convex bodies succinctly by convex polytopes is a fundamental problem in discrete geometry. A convex body $K$ of diameter $\diam(K)$ is given in Euclidean $d$-dimensional space, where $d$ is a constant. Given an error parameter $\eps > 0$, the objective is to determine a polytope of minimum combinatorial complexity whose Hausdorff distance from $K$ is at most $\eps \cdot \diam(K)$. By combinatorial complexity we mean the total number of faces of all dimensions of the polytope. A well-known result by Dudley implies that $O(1/\eps^{(d-1)/2})$ facets suffice, and a dual result by Bronshteyn and Ivanov similarly bounds the number of vertices, but neither result bounds the total combinatorial complexity. We show that there exists an approximating polytope whose total combinatorial complexity is $\tilde{O}(1/\eps^{(d-1)/2})$, where $\tilde{O}$ conceals a polylogarithmic factor in $1/\eps$. This is a significant improvement upon the best known bound, which is roughly $O(1/\eps^{d-2})$.

Our result is based on a novel combination of both old and new ideas. First, we employ Macbeath regions, a classical structure from the theory of convexity. The construction of our approximating polytope employs a new stratified placement of these regions. Second, in order to analyze the combinatorial complexity of the approximating polytope, we present a tight analysis of a width-based variant of B\'{a}r\'{a}ny and Larman's \emph{economical cap covering}. Finally, we use a deterministic adaptation of the witness-collector technique (developed recently by Devillers {\etal}) in the context of our stratified construction.
\end{abstract}

\bigskip
\noindent\textbf{Keywords:} Convex polytopes, polytope approximation, combinatorial complexity, Macbeath regions
\bigskip

\section{Introduction} \label{s:intro}

Approximating general convex bodies by convex polytopes is a fundamental geometric problem. It has been extensively studied in the literature under various formulations. (See Bronstein~\cite{Bro08} for a survey.) Consider a convex body $K$, that is, a closed, convex set of bounded diameter, in Euclidean $d$-dimensional space. At issue is the structure of the simplest polytope $P$ that approximates $K$.

There are various ways to define the notions of ``simplest'' and ``approximates.'' Our notion of approximation will be based on the \emph{Hausdorff metric}, that is, the maximum distance between a point in the boundary of $P$ or $K$ and the boundary of the other body. Normally, approximation error is defined relative to $K$'s diameter. It will simplify matters to assume that $K$ has been uniformly scaled to unit diameter. For a given error $\eps > 0$, we say that a polytope $P$ is an \emph{$\eps$-approximating polytope} to $K$ if the Hausdorff distance between $K$ and $P$ is at most $\eps$. The simplicity of an approximating polytope $P$ will be measured in terms of its \emph{combinatorial complexity}, that is, the total number of $k$-faces, for $0 \le k \le d-1$. For the purposes of stating asymptotic bounds, we assume that the dimension $d$ is a constant.

The bounds given in the literature for convex approximation are of two common types~\cite{Bro08}. In both cases, the bounds hold for all $\eps \leq \eps_0$, for some $\eps_0 > 0$. In \emph{nonuniform bounds}, the value of $\eps_0$ depends on $K$ (for example, on $K$'s maximum curvature). Such bounds are often stated as holding ``in the limit'' as $\eps$ approaches zero, or equivalently as the combinatorial complexity of the approximating polytope approaches infinity. Examples include bounds by Gruber~\cite{Gru93}, Clarkson~\cite{Cla06}, and others 
\cite{Bor00,Sch87,Tot48}. Our interest is in \emph{uniform bounds}, where the value of $\eps_0$ is independent of $K$. Examples include the results of Dudley~\cite{Dud74} and Bronshteyn and Ivanov~\cite{BrI76}. Such bounds hold without any assumptions on $K$.

Dudley showed that, for $\eps \leq 1$, any convex body $K$ of unit diameter can be $\eps$-approximated by a convex polytope $P$ with $O(1/\eps^{(d-1)/2})$ facets. This bound is known to be optimal in the worst case and is achieved when $K$ is a Euclidean ball (see, e.g.,~\cite{Bro08}). Alternatively, Bronshteyn and Ivanov showed the same bound holds for the number of vertices, which is also the best possible. No convex polytope approximation is known that attains both bounds simultaneously.%
\footnote{Jeff Erickson noted that both bounds can be attained simultaneously but at the cost of sacrificing convexity~\cite{Cla06}.}

Establishing good uniform bounds on the combinatorial complexity of convex polytope approximations is a major open problem. The Upper-Bound Theorem~\cite{MCM70} implies that a polytope with $n$ vertices (resp., facets) has total combinatorial complexity $O(n^{\floor{d/2}})$. Applying this to the results of either Dudley or Bronshteyn and Ivanov directly yields a bound of $O(1/\eps^{(d^2-d)/4})$ on the combinatorial complexity of an $\eps$-approximating polytope. Better uniform bounds without $d^2$ in the exponent are known, however. Consider a uniform grid $\Psi$ of points with spacing $\Theta(\eps)$, and let $P$ denote the convex hull of $\Psi \cap K$. It is easy to see that $P$ is an $\eps$-approximating polytope for $K$. The combinatorial complexity of any lattice polytope%
\footnote{A \emph{lattice polytope} is the convex hull of any set of points with integer coordinates.}
is known to be $O(V^{(d-1)/(d+1)})$, where $V$ is the volume of the polytope~\cite{And63,Bar08}. This implies that $P$ has combinatorial complexity $O(1/\eps^{d(d-1)/(d+1)}) \approx O(1/\eps^{d - 2})$. While this is significantly better than the bound provided by the Upper-Bound Theorem, it is still much larger than the lower bound of $\Omega(1/\eps^{(d-1)/2})$.

We show that this gap can be dramatically reduced. In particular, we establish an upper bound on the combinatorial complexity of convex approximation that is optimal up to a polylogarithmic factor in $1/\eps$.

\begin{theorem} \label{thm:main}
Let $K \subset \RE^d$ be a convex body of unit diameter, where $d$ is a fixed constant. For all sufficiently small positive $\eps$ (independent of $K$) there exists an $\eps$-approximating convex polytope $P$ to $K$ of combinatorial complexity $O(1/{\widehat{\eps}^{\kern+1pt (d-1)/2}})$, where $\widehat{\eps} = \eps/\log (1/\eps)$.
\end{theorem}

This is within a factor of $O(\log^{(d-1)/2} (1/\eps))$ of the aforementioned lower bound. Our approach employs a classical structure from the theory of convexity, called \emph{Macbeath regions}~\cite{Mac52}. Macbeath regions have found numerous uses in the theory of convex sets and the geometry of numbers (see B\'{a}r\'{a}ny~\cite{Bar00} for an excellent survey). They have also been applied to a small but growing number of results in the field of computational geometry (see, e.g.,~\cite{BCP93, AMX12, AMM09b, AFM12b}). Our construction of the approximating polytope uses a new stratified placement of these regions. In order to analyze the combinatorial complexity of the approximating polytope, in Section~\ref{s:ecc} we present a tight analysis of a width-based variant of B\'{a}r\'{a}ny and Larman's economical cap covering. This result plays a central role in our recent work on approximate polytope membership queries \cite{AFM17a} and may find use in other applications. Finally, we employ a deterministic version of the witness-collector technique, developed recently by Devillers {\etal}~\cite{DGG13}, in the context of our stratified construction.

The paper is organized as follows. In Section~\ref{s:prelim}, we define concepts related to Macbeath regions and present some of their key properties. In Section~\ref{s:ecc}, we prove the width-based economical cap covering lemma. The stratified placement of the Macbeath regions and the bound on the combinatorial complexity of approximating polytopes follow in Section~\ref{s:approx}. We conclude with several open problems in Section~\ref{s:conclusion}.

\section{Geometric Preliminaries} \label{s:prelim}

Recall that $K$ is a convex body of unit diameter in $\RE^d$. Let $\partial K$ denote its boundary. Let $O$ denote the origin of $\RE^d$, and for $x \in \RE^d$ and $r \ge 0$, let $B^r(x)$ denote the Euclidean ball of radius $r$ centered at $x$. It will be convenient to first map $K$ to a convenient form. We say that a convex body $K$ is in \emph{canonical form} if $B^{1/2d}(O) \subseteq K \subseteq B^{1/2}(O)$. Given a parameter $0 < \gamma \le 1$, we say that a convex body $K$ is \emph{$\gamma$-fat} if there exist concentric Euclidean balls $B$ and $B'$, such that $B \subseteq K \subseteq B'$, and $\radius(B) / \radius(B') \ge \gamma$. Thus, a body in canonical form is $(1/d)$-fat and has diameter $\Theta(1)$. We will refer to point $O$ as the \emph{center} of $K$.

The following lemma shows that, up to constant factors, the problem of approximating an arbitrary convex body can be reduced to approximating a convex body in canonical form. The proof follows from a combination of John's Theorem~\cite{John} and Lemma~{3.1} of Agarwal {\etal}~\cite{AHV04} and is included for completeness.

\begin{lemma} \label{lem:canonical}
Let $K$ be a convex body of unit diameter in $\RE^d$. There exists a non-singular affine transformation $T$ such that $T(K)$ is in canonical form and if $P$ is any $(\eps/d)$-approximating polytope to $T(K)$, then $T^{-1}(P)$ is an $\eps$-approximating polytope to $K$.
\end{lemma}

\begin{proof}
Let $E$ denote a maximum volume ellipsoid enclosed within $K$ (that is, the John ellipsoid). Since $K$ is of unit diameter, $E$'s semi-principal axes are all of length at most $1/2$. Consider a frame centered at $E$'s center and whose axes coincide with $E$'s semi-principal axes. Let $T$ be an affine transformation that maps this frame's origin to the origin of the space, and scales all of the frame's basis vectors to length $1/2d$. This affine transformation maps $E$ to $B^{1/2d}(O)$. Since each of the frame's basis vectors is scaled from a length of at most $1/2$ to a length of $1/2d$, it follows that $T$ maps any vector $v$ to a vector of length at least $\|v\|/d$. Thus, $T^{-1}$ maps any vector $v$ to a vector of length at most $d \|v\|$. Therefore, if $P$ is any $(\eps/d)$-approximating polytope to $T(K)$, $T^{-1}(P)$ is an $\eps$-approximating polytope to $T^{-1}(T(K)) = K$, as desired.
\end{proof}

We assume henceforth that $K$ is given in canonical form and that $\eps$ has been appropriately scaled. This scaling only affects the constant factors in our asymptotic bounds.

A \emph{cap} $C$ is defined to be the nonempty intersection of the convex body $K$ with a halfspace $H$ (see Figure~\ref{f:cap}(a)). Let $h$ denote the hyperplane bounding $H$. We define the \emph{base} of $C$ to be $h \cap K$. The \emph{apex} of $C$ is any point in the cap such that the supporting hyperplane of $K$ at this point is parallel to $h$. The \emph{width} of $C$ is the distance between $h$ and this supporting hyperplane. Given any cap $C$ of width $w$ and a real $\lambda \ge 0$, we define its \emph{$\lambda$-expansion}, denoted $C^{\lambda}$, to be the cap of $K$ cut by a hyperplane parallel to and at distance $\lambda w$ from this supporting hyperplane. (Note that $C^{\lambda} = K$, if $\lambda w$ exceeds the width of $K$ along the defining direction.) An easy consequence of convexity is that, for $\lambda \ge 1$, $C^{\lambda}$ is a subset of the region obtained by scaling $C$ by a factor of $\lambda$ about its apex. It follows that, for $\lambda \ge 1$, $\vol(C^{\lambda}) \le \lambda^d \cdot \vol(C)$. For a given $\eps > 0$, let $K(\eps) \subset K$ denote the points of $K$ within distance at most $\eps$ from $\partial K$ (equivalently, the union of all $\eps$-width caps).

\begin{figure}[htbp]
	\centerline{\includegraphics[scale=.75]{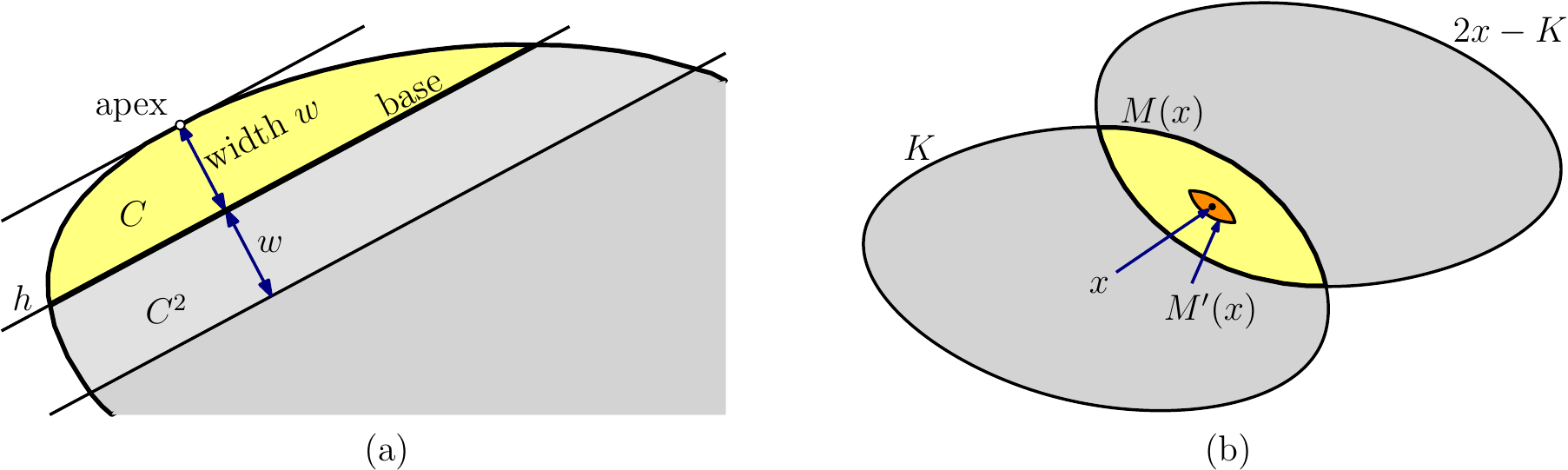}}
	\caption{\label{f:cap}(a) Cap concepts and (b) Macbeath regions.}
\end{figure}

Given a point $x \in K$ and real parameter $\lambda \ge 0$, the \emph{Macbeath region} $M^{\lambda}(x)$ (also called an \emph{M-region}) is defined as:
\[
  M^{\lambda}(x) ~ = ~ x + \lambda ((K-x) \cap (x-K)).
\]
It is easy to see that $M^{1}(x)$ is the intersection of $K$ and the reflection of $K$ around $x$ (see Figure~\ref{f:cap}(b)), and so $M^{1}(x)$ is centrally symmetric about $x$. $M^{\lambda}(x)$ is a scaled copy of $M^{1}(x)$ by the factor $\lambda$ about $x$. We refer to $x$ as the \emph{center} of $M^{\lambda}(x)$ and to $\lambda$ as its \emph{scaling factor}. As a convenience, we define $M(x) = M^1(x)$ and $M'(x) = M^{1/5}(x)$.

We begin with two lemmas that encapsulate relevant properties of Macbeath regions. Both were proved originally by Ewald, Larman, and Rogers~\cite{ELR70}, but our statements follow the forms given by Br\"{o}nnimann, Chazelle, and Pach~\cite{BCP93}. (Lemmas~\ref{lem:mac-mac} and~\ref{lem:mac2} below are restatements of Lemmas~{2.5} and~{2.6} from~\cite{BCP93}, respectively.)

\begin{lemma} \label{lem:mac-mac}
Let $K$ be a convex body. If $x, y \in K$ such that $M'(x) \cap M'(y) \neq \emptyset$, then $M'(y) \subseteq M(x)$. 
\end{lemma}

\begin{lemma} \label{lem:mac2}
Let $K \subset \RE^d$ be a convex body in canonical form, and let $\Delta_0 = 1/(6d)$ be a constant. Let $C$ be a cap of $K$ of width at most $\Delta_0$.  Let $x$ denote the centroid of the base of this cap. Then $C \subseteq M^{3d}(x)$.
\end{lemma}

The following lemma is an immediate consequence of the definition of Macbeath region.

\begin{lemma} \label{lem:mac1}
Let $K$ be a convex body and $\lambda > 0$. If $x$ is a point in a cap $C$ of $K$, then $M^{\lambda}(x) \cap K \subseteq C^{1+\lambda}$. Furthermore, if $\lambda \le 1$, then $M^{\lambda}(x) \subseteq C^{1+\lambda}$.
\end{lemma}

The next lemma is useful in situations when we know that a Macbeath region partially overlaps a cap of $K$. It allows us to conclude that a constant factor expansion of the cap will fully contain the Macbeath region.

\begin{lemma} \label{lem:cap-mac}
Let $K$ be a convex body. Let $C$ be a cap of $K$ and $x$ be a point in $K$ such that $C \cap M'(x) \neq \emptyset$. Then $M'(x) \subseteq C^2$. 
\end{lemma}

\begin{proof}
Let $y$ be any point in $C \cap M'(x)$. Since $M'(x) \cap M'(y) \neq \emptyset$ obviously holds, we can apply Lemma~\ref{lem:mac-mac} to conclude that $M'(x) \subseteq M(y)$. By Lemma~\ref{lem:mac1} (with $\lambda = 1$), $M(y) \subseteq C^2$. It follows that $M'(x) \subseteq C^2$.
\end{proof}

Next, we give two straightforward lemmas dealing with scaling of centrally symmetric convex bodies. As Macbeath regions are centrally symmetric, these lemmas will be useful to us in conjunction with their standard properties. A proof of Lemma~\ref{lem:sym1} appears in B\'{a}r\'{a}ny~\cite{Bar89}. For any centrally symmetric convex body $A$, define $A^{\lambda}$ to be the body obtained by scaling $A$ by a factor of $\lambda$ about its center. 

\begin{lemma} \label{lem:sym1}
Let $\lambda \ge 1$. Let $A$ and $B$ be centrally symmetric convex bodies such that $A \subseteq B$. Then $A^\lambda \subseteq B^\lambda$.
\end{lemma}

\begin{lemma} \label{lem:sym2}
Let $\lambda \ge 1$. Let $A$ be a centrally symmetric convex body. Let $A'$ be the body obtained by scaling $A$ by a factor of $\lambda$ about any point in $A$. Then $A' \subseteq A^{2\lambda-1}$.
\end{lemma}

\begin{proof}
We take the origin to be at the center of $A$. Let $A'$ be the body obtained by scaling $A$ by a factor of $\lambda$ about a point $a \in A$. Any point $u$ in $A'$ is of the form $a + \lambda (x-a)$, where $x \in A$. This can be expressed as
\[
(2 \lambda - 1) \left[ \frac{\lambda}{2 \lambda - 1} x + \frac{\lambda - 1}{2 \lambda - 1} (-a) \right].
\]
Since $\lambda \ge 1$, the point $(\lambda/(2\lambda - 1)) x + ((\lambda-1)/(2\lambda-1)) (-a)$ lies on the segment joining $x$ and $-a$. Since both $x$ and $-a$ lie within $A$, it follows that $u \in A^{2\lambda-1}$, as desired.
\end{proof}

The following lemma is an easy consequence of Lemmas~\ref{lem:mac2} and \ref{lem:sym2}.

\begin{lemma} \label{lem:mac3}
Let $\lambda \ge 1$ and let $K, C$, and $x$ be as defined in Lemma~\ref{lem:mac2}.
Then $C^\lambda \subseteq M^{3d(2\lambda-1)}(x)$.
\end{lemma}

\begin{proof}
By Lemma~\ref{lem:mac2}, $C \subseteq M^{3d}(x)$. Recall that $C^\lambda$ is contained within the region obtained by scaling $C$ by a factor of $\lambda$ about its apex. Applying Lemma~\ref{lem:sym2} (applied to $M^{3d}(x)$ and the apex point), it follows that $C^{\lambda} \subseteq M^{3d(2\lambda-1)}(x)$.
\end{proof}

The well known Lemma~\ref{lem:mac-mac} states that if two $(1/5)$-shrunken Macbeath regions have a nonempty intersection, then a constant factor expansion of one contains the other~\cite{ELR70,BCP93}. We show next that this holds for the associated caps as well. (Note that this does not hold in general for overlapping caps. If two caps $C_1$ and $C_2$ have a nonempty intersection, there is no constant $\beta$ that guarantees that $C_1 \subseteq C_2^{\beta}$.)

\begin{lemma} \label{lem:cap-cap}
Let $\Delta_0$ be the constant of Lemma~\ref{lem:mac2} and let $\lambda \ge 1$ be any real. There exists a constant $\beta \ge 1$ such that the following holds. Let $K \subset \RE^d$ be a convex body in canonical form. Let $C_1$ and $C_2$ be any two caps of $K$ of width at most $\Delta_0$. Let $x_1$ and $x_2$ denote the centroids of the bases of the caps $C_1$ and $C_2$, respectively. If $M'(x_1) \cap M'(x_2) \neq \emptyset$, then $C_1^{\lambda} \subseteq C_2^{\beta \lambda}$.
\end{lemma}

\begin{proof}
By Lemma~\ref{lem:mac3}, $C_1^{\lambda} \subseteq M^{\alpha}(x_1)$, where $\alpha = 3d(2\lambda-1)$. Since $M'(x_1)$ and $M'(x_2)$ overlap, by Lemma~\ref{lem:mac-mac}, $M'(x_1) \subseteq M(x_2)$. By definition, $M'(x_1) = M^{1/5}(x_1)$ and so $M^{\alpha}(x_1) = (M'(x_1))^{5\alpha}$. Since $M'(x_1)$ and $M(x_2)$ are centrally symmetric bodies and $M'(x_1) \subseteq M(x_2)$, by Lemma~\ref{lem:sym1}, it follows that $(M'(x_1))^{5\alpha} \subseteq M^{5\alpha}(x_2)$. Putting it together, we obtain
\[
  C_1^{\lambda} 
	~ \subseteq ~ M^{\alpha}(x_1) 
	~     =     ~ (M'(x_1))^{5\alpha} 
	~ \subseteq ~ M^{5\alpha}(x_2).
\]
By Lemma~\ref{lem:mac1}, $M^{5\alpha}(x_2) \cap K \subseteq C_2^{1+5\alpha}$. Since $C_1^{\lambda} \subseteq M^{5\alpha}(x_2)$ and $C_1^{\lambda} \subseteq K$, we have $C_1^{\lambda} \subseteq M^{5\alpha}(x_2) \cap K \subseteq C_2^{1+5\alpha}$. Recalling that $\alpha = 3d(2\lambda-1)$, we have $C_1^{\lambda} \subseteq C_2^{30d\lambda}$. This proves the lemma for constant $\beta = 30d$.
\end{proof}

\section{Economical Cap Covering} \label{s:ecc}

In this section we present a tight analysis of a width-based variant of B\'{a}r\'{a}ny and Larman's economical cap covering~\cite{BaL88}. The lemma applies generally to any convex body $K$ that has constant diameter and is $\gamma$-fat for some constant $\gamma$ (where the constants may depend on $d$). The proof of this lemma follows from the ideas in~\cite{ELR70, BaL88,Bar89}. Our principal contribution is an optimal bound of $O(1/\eps^{(d-1)/2})$ on the number of bodies needed. 

\begin{lemma}[Width-based economical cap covering lemma] \label{lem:ecc}
Let $\eps > 0$ be a sufficiently small parameter. Let $K \subset \RE^d$ be a convex body in canonical form. There exists a collection $\RR$ of $k = O(1/\eps^{(d-1)/2})$ disjoint centrally symmetric convex bodies $R_1, \ldots, R_k$ (see Figure~\ref{f:cap-cover}(a)) and associated caps $C_1,\ldots, C_k$ such that the following hold (for some constants $\beta$ and $\lambda$, which depend only on $d$):
\begin{enumerate}
\item For each $i$, $C_i$ is a cap of width $\beta \eps$, and $R_i \subseteq C_i \subseteq R_i^\lambda$.

\item Let $C$ be any cap of width $\eps$. Then there is an $i$ such that $R_i \subseteq C$ and $C_i^{1/\beta^2} \subseteq C \subseteq C_i$ (see Figure~\ref{f:cap-cover}(b)).
\end{enumerate}
\end{lemma}

\begin{figure}[htbp]
	\centerline{\includegraphics[scale=.75]{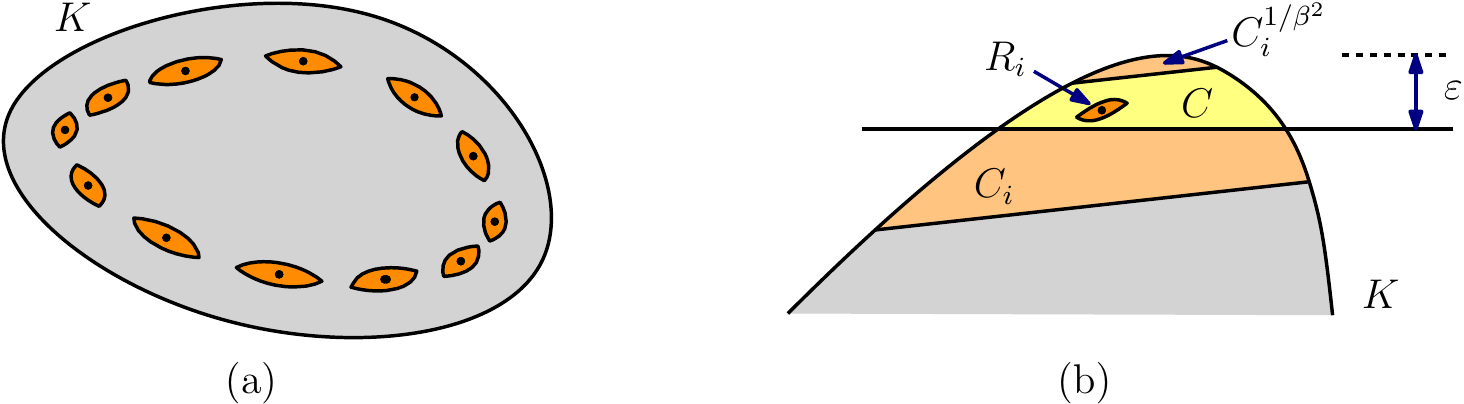}}
	\caption{\label{f:cap-cover}Illustrating Lemma~\ref{lem:ecc}.}
\end{figure}

The $R_i$'s in this lemma are Macbeath regions with scaling factor $1/5$. Since any cap of width $\eps$ is contained in some cap $C_i$, it follows that the $C_i$'s together cover $K(\eps)$. Further, from Property~1, we can see that the sum of the volume of the $C_i$'s is no more than a constant times the volume of $K(\eps)$. It is in this sense that the $C_i$'s constitute an \emph{economical} cap covering.

It is worth mentioning that Property~2 is stronger than similar properties given previously in the literature in the following sense. For any cap of width $\eps$, we show not merely that it is contained within some cap $C_i$ of the cover, but it is effectively ``sandwiched'' between two caps with parallel bases, each of width $\Theta(\eps)$.

A key technical contribution of our paper is the following lemma. It will help us bound the number of bodies needed in the width-based cap covering lemma. Because of its broader utility, this lemma is given in a slightly more general form than is needed here.

\begin{lemma} \label{lem:bound-mac}
Let $K \subset \RE^d$ be a convex body in canonical form. Let $0 < \delta \le \Delta_0/2$, where $\Delta_0$ is the constant of Lemma~\ref{lem:mac2}. Let $\mathcal{C}$ be a set of caps, whose widths lie between $\delta$ and $2\delta$, such that the Macbeath regions $M'(x)$ centered at the centroids $x$ of the bases of these caps are disjoint. Then $|\mathcal{C}| = O(1/\delta^{(d-1)/2})$.
\end{lemma}

Our proof of Lemma~\ref{lem:bound-mac} will require the following geometric observation, which is a straightforward extension of Dudley's convex approximation construction (see Lemma~{4.4} of~\cite{Dud74}). It is similar to other results based on Dudley's construction (including Lemma~{3.6} of~\cite{AHV04} and Lemma~{23.12} of~\cite{SarielBook}). We will present the proof for the sake of completeness. Let $S$ denote the sphere of radius 2 centered at the origin $O$, which we call the \emph{Dudley sphere}. Given vectors $u$ and $v$, let $\ang{u,v}$ denote their dot product and let $\|u\| = \ang{u,u}^{1/2}$ denote $u$'s Euclidean length.

\begin{lemma} \label{lem:dudley}
Let $K$ be a convex body that lies within a unit sphere centered at the origin, and let $0 < \delta \le 1$. Let $x'$ and $y'$ be two points of $S$. Let $x$ and $y$ be the points of $\partial K$ that are closest to $x'$ and $y'$, respectively. Let $h$ denote the supporting hyperplane at $x$ orthogonal to the segment $x x'$. Let $C$ denote the cap cut from $K$ by a hyperplane parallel to and at distance $\delta$ from $h$. If $y \notin C$, then $\|x' - y'\| \ge \sqrt{\delta}$.
\end{lemma}

\begin{proof}
Before starting the proof, we recall a technical result (Lemma~{4.3}) from Dudley~\cite{Dud74}, which states that given vectors $x$, $y$, $u$, $v$ in $\RE^d$ such that $\ang{x - y, u} \ge 0$ and $\ang{x - y, v} \le 0$, $\|(x + u) - (y + v)\| \ge \max(\|x - y\|, \|u - v\|)$. This follows from the observation that 
\[
	\|(x + u) - (y + v)\|^2
		~  =  ~ \|x - y\|^2 + \|u - v\|^2 + 2 \ang{x - y, u - v} 
		~ \ge ~ \|x - y\|^2 + \|u - v\|^2.
\]

Returning to the proof, suppose towards a contradiction that $y \notin C$ but $\|x'-y'\| < \sqrt{\delta}$. Let $u = x'-x$ and $v = y'-y$, and let $\widehat{u} = u/\|u\|$ and $\widehat{v} = v/\|v\|$ (see Figure~\ref{f:dudley}). Clearly, $\|(x + u) - (y + v)\| = \|x' - y'\| < \sqrt{\delta}$. A direct consequence of convexity is that $\ang{x - y, u} \ge 0$ and $\ang{x - y, v} \le 0$, and so by the above result it follows that $\|x - y\|$ and $\|u - v\|$ are both less than $\sqrt{\delta}$. Clearly, $u$ and $v$ are of at least unit length, and thus $\|\widehat{u} - \widehat{v}\| \le \|u - v\| < \sqrt{\delta}$. Let $\theta$ denote the angle between $\widehat{u}$ and $\widehat{v}$. Since $\|x' - y'\| < \sqrt{\delta} \le 1$ and the radius of $S$ is 2, it follows that $\theta < \pi/2$.

\begin{figure}[htbp]
	\centerline{\includegraphics[scale=.75]{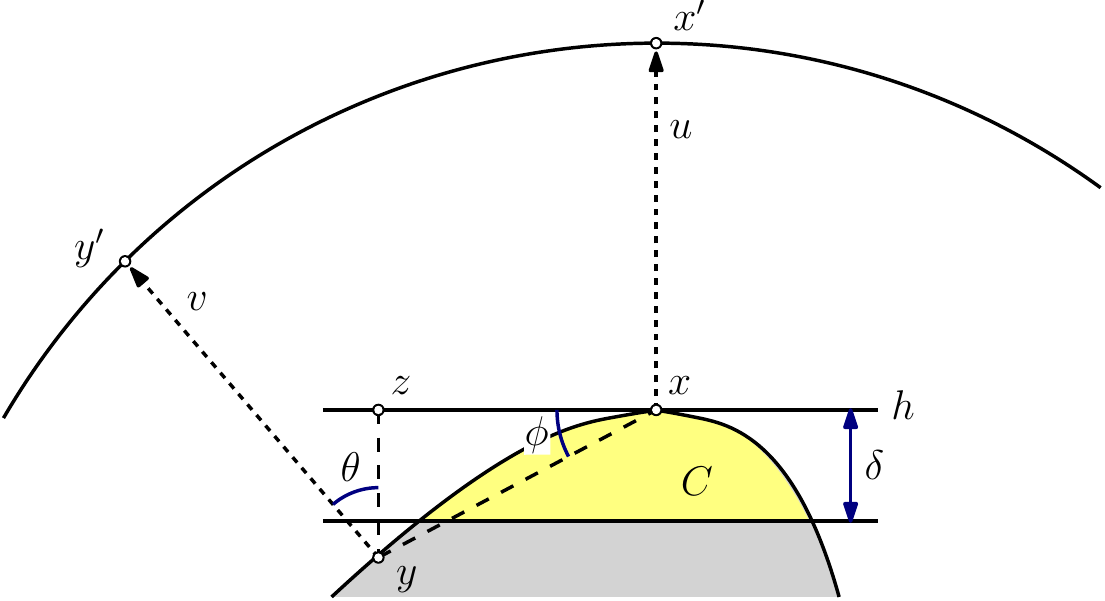}}
	\caption{\label{f:dudley}Illustrating Lemma~\ref{lem:dudley}.}
\end{figure}

Consider the right triangle whose hypotenuse is $x y$ and whose third vertex is the orthogonal projection of $y$ onto the supporting hyperplane $h$, which we denote by $z$. Letting $\phi = \angle z x y$, it follows from convexity that $\phi \le \theta$. (This is because any supporting hyperplane through $y$ cannot pass below $x$.) Because $\theta < \pi/2$, $\sin \theta \ge \sin \phi$. Also, since $y \notin C$, we have $\|z - y\| > \delta$, and therefore 
\[
	\sin \theta 
		~ \ge ~ \sin \phi
		~  =  ~ \frac{\|z - y\|}{\|x - y\|} 
		~  >  ~ \frac{\delta}{\sqrt{\delta}}
		~  =  ~ \sqrt{\delta}. 
\]
Observe that $\|\widehat{u} - \widehat{v}\|$ is the length of a chord of a unit circle that subtends an arc of angle $\theta$, and therefore $\|\widehat{u} - \widehat{v}\| =  2 \sin \frac{\theta}{2}$. Given our earlier bound on this distance, we obtain the following contradiction:
\[
	\sqrt{\delta} 
		~  <  ~ \sin \theta 
		~  =  ~ 2 \sin \frac{\theta}{2} \cos \frac{\theta}{2}
		~ \le ~ 2 \sin \frac{\theta}{2}
                ~  =  ~ \|\widehat{u} - \widehat{v}\|
		~  <  ~ \sqrt{\delta}.
\]
\end{proof}
 
We are now ready to present the proof of Lemma~\ref{lem:bound-mac}.

\begin{proof}(of Lemma~\ref{lem:bound-mac})
Let $A$ be the set of disjoint Macbeath regions $M'(x)$ described in the lemma. For each region $M'(x)$, let $C(x)$ denote the cap whose base centroid point generates $M'(x)$. We begin by pruning $A$ to obtain a subset $B$, which to within constant factors has the same cardinality as $A$. We construct $B$ incrementally as follows. Initially $B$ is the empty set. In each step, from among the Macbeath regions that still remain in $A$, we choose a Macbeath region $M'(x)$ that has the smallest volume, and insert it into $B$. We then prune all the Macbeath regions from $A$ that intersect the cap $C^4(x)$. We continue in this manner until $A$ is exhausted.

We claim that in each step, we prune a constant number of Macbeath regions from $A$. Let $M'(x)$ denote the Macbeath region inserted into $B$ in this step. If $M'(y)$ is a Macbeath region that is pruned in this step, then $M'(y)$ intersects the cap $C^4(x)$. It then follows from Lemma~\ref{lem:cap-mac} that $M'(y) \subseteq C^8(x)$. Note that
\[
  \vol(C^8(x)) 
	~ \le ~ 8^d \vol(C(x)) 
	~  =  ~ O(\vol(C(x))). 
\]
Since $C(x)$ is of width at most $2 \delta \le \Delta_0$, we may apply Lemma~\ref{lem:mac2}, which yields $C(x) \subseteq M^{3d}(x)$. It follows that 
\[
  \vol(M(x)) 
	~ \ge ~ \vol(C(x)) / (3d)^d 
	~  =  ~ \Omega(\vol(C(x))). 
\]
Recall that each Macbeath region pruned has volume greater than or equal to the volume of $M'(x)$. It follows that the volume of each Macbeath region pruned is $\Omega(\vol(M(x))) = \Omega(\vol(C(x)))$. Since the pruned Macbeath regions are disjoint and contained in a region of volume $O(\vol(C(x)))$, a straightforward packing argument implies that the number of Macbeath regions pruned is $O(1)$.

The claim immediately implies that $|A| = O(|B|)$. In the remainder of the proof, we will show that $|B| = O(1/\delta^{(d-1)/2})$, which will complete the proof.

Let $X$ denote the set of centers of the Macbeath regions of $B$, that is, $X = \{x: M'(x) \in B\}$.  We map each point $x \in X$ to a point $x'$ on the Dudley sphere such that $xx'$ is normal to the base of the cap $C(x)$. We claim that the distance between any pair of the projected points $x'$ on the Dudley sphere is at least $\sqrt{\delta}$. Note that this claim would imply the desired bound on $|B|$ and complete the proof.

To see this claim, consider any two Macbeath regions $M'(x)$ and $M'(y)$ in the set $B$. Without loss of generality, suppose that $M'(y)$ is inserted into $B$ after $M'(x)$. By our construction, it follows that $y$ is not contained in $C^4(x)$ (because otherwise $M'(y)$ would intersect $C^4(x)$ and would have been pruned after inserting $M'(x)$ into $B$). We now consider two cases, depending on whether or not $x$ is contained in $C(y)$. 

\medskip\noindent\textbf{Case 1:}
($x \notin C(y)$)
Consider the convex body $K'$ that is the closure of $K \setminus (C(x) \cup C(y))$ (outlined in red in Figure~\ref{f:case12}(a)). Note that $x$ and $y$ are on the boundary of the convex body $K'$ and these are the points of $\partial K'$ that are closest to $x'$ and $y'$, respectively. Next, consider the cap of $K'$ whose apex is $x$ and width is $\delta$. Call this cap $C'(x)$. Since the width of $C(x)$ is at least $\delta$, and $y \notin C^4(x)$, it is easy to see that $y \notin C'(x)$. Applying Lemma~\ref{lem:dudley} to the convex body $K'$ and the points $x'$, $y'$, $x$, and $y$, it follows that $\|x'y'\| \geq \sqrt{\delta}$.

\begin{figure}[htbp]
	\centerline{\includegraphics[scale=.75]{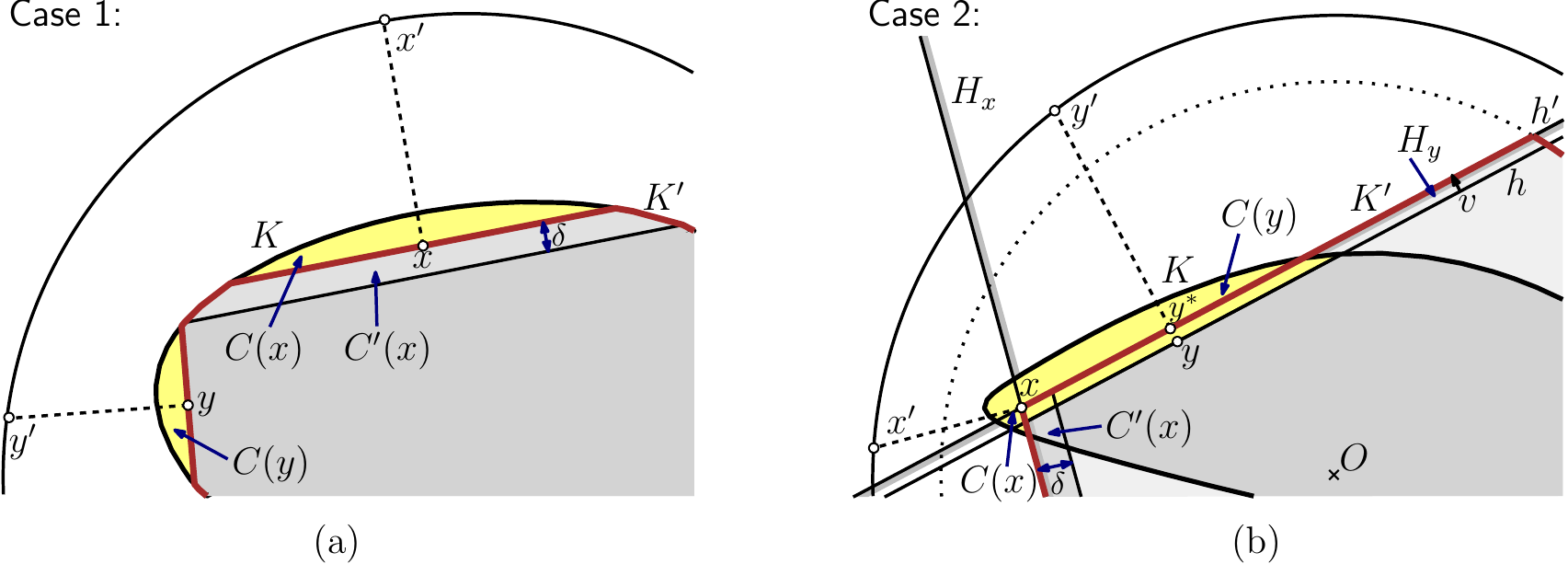}}
	\caption{\label{f:case12}Cases arising in the proof of Lemma~\ref{lem:bound-mac}. (Figure not to scale.)}
\end{figure}

\medskip\noindent\textbf{Case 2:}
($x \in C(y)$)
Let $h$ denote the hyperplane that forms the base of $C(y)$ (see Figure~\ref{f:case12}(b)). Let $h'$ denote the hyperplane parallel to $h$ that passes through $x$. Let $v$ denote the vector normal to $h$, whose magnitude is the distance between $h$ and $h'$. Note that $h'$ = $h$ + $v$. Since $C(y)$ is a cap of width at most $2 \delta$, the magnitude of the translation vector $v$ is at most $2 \delta$. Let $y^* = y + v$. Let $H_y$ denote the halfspace bounded by $h'$ that contains the origin. Let $H_x$ denote the halfspace that contains the origin and whose boundary is the hyperplane forming the base of $C(x)$. Define the convex body $K'$ as the intersection of $H_x$ and $H_y$ and a ball of unit radius centered at the origin. Note that $x$ and $y^*$ lie on the boundary of $K'$ (since $\|Ox\| < 1$ and $\|Oy^*\| < 1$; $\|Ox\| < 1$ holds trivially since $x \in K$ and $K \subseteq B^{1/2}(O)$, and $\|Oy^*\| \le \|Oy\| + \|yy^*\| \le 1/2 + 2\delta \le 1/2 + 2\Delta_0 < 1$). 

Further, the points $x$ and $y^*$ are the points of $\partial K'$ that are closest to $x'$ and $y'$, respectively. Next, consider the cap of $K'$ whose apex is $x$ and width is $\delta$ and whose base is parallel to the base of $C(x)$. Call this cap $C'(x)$. Recall that $y \notin C^4(x)$, the width of $C(x)$ is at least $\delta$, and the distance between $y$ and $y^*$ is at most $2\delta$. It follows that $y^*$ is at distance bigger than $3\delta - 2\delta = \delta$ from the hyperplane passing through the base of $C(x)$. Since the distance between the hyperplanes passing through the bases of $C(x)$ and $C'(x)$, respectively, is $\delta$, it follows that $y^* \notin C'(x)$. Applying Lemma~\ref{lem:dudley} to the convex body $K'$ and the points $x', y', x$, and $y^*$, it follows that the distance between $x'$ and $y'$ is at least $\sqrt{\delta}$.
This establishes the above claim and completes the proof.
\end{proof}

The remainder of this section is devoted to proving Lemma~\ref{lem:ecc}.

\begin{proof}
Assume that $\eps \le \Delta_0$, where $\Delta_0$ is the constant of Lemma~\ref{lem:mac2}. Let $\beta = 30d$ be the constant of Lemma~\ref{lem:cap-cap}. Let $\mathcal{C}$ be a maximal set of caps, each of width $\eps/\beta$, such that the $(1/5)$-scaled Macbeath regions centered at the centroids of the bases of these caps are disjoint. Let $A_1,\ldots,A_k$ denote the caps of $\mathcal{C}$. Let $x_i$ denote the centroid of the base of cap $A_i$.  With each cap $A_i$, we associate a convex body $R_i = M'(x_i)$ and a cap $C_i = A_i^{\beta^2}$. We will show that the convex bodies $R_i$ and caps $C_i$ satisfy the properties given in the lemma. 

By Lemma~\ref{lem:bound-mac}, $|\mathcal{C}| = O(1/\eps^{(d-1)/2})$, which implies the desired upper bound on $k$. Since $C_i$ is a $\beta^2$-expansion of $A_i$, its width is $\beta \eps$. To prove Property~1, it remains to show that $M'(x_i) \subseteq C_i \subseteq (M'(x_i))^{\lambda}$. By Lemma~\ref{lem:mac1}, $M'(x_i) \subseteq A_i^{6/5}$. Since $A_i^{6/5}  \subseteq A_i^{\beta^2} = C_i$, we obtain $M'(x_i) \subseteq C_i$. Also, applying Lemma~\ref{lem:mac3}, we obtain 
\[
  C_i 
	~     =     ~ A_i^{\beta^2} 
	~ \subseteq ~ M^{3d(2\beta^2-1)}(x_i) 
    ~     =     ~ (M'(x_i))^{15d(2\beta^2-1)} 
	~ \subseteq ~ (M'(x_i))^{\lambda},
\]
where $\lambda = 30 d \beta^2$. Thus, $M'(x_i) \subseteq C_i \subseteq (M'(x_i))^\lambda$.

To show Property~2, let $C$ be any cap of width $\eps$. Let $x$ denote the centroid of the base of $C^{1/\beta}$. By maximality of $\CC$, there must be a Macbeath region $M'(x_i)$ that has a nonempty intersection with $M'(x)$ (note $x_i$ may be the same as point $x$). Applying Lemma~\ref{lem:mac-mac}, it follows that $M'(x_i) \subseteq M(x)$. By Lemma~\ref{lem:mac1}, $M(x) \subseteq C^{2/\beta}$. Putting it together, we obtain $M'(x_i) \subseteq M(x) \subseteq C^{2/\beta} \subseteq C$, which establishes the first part of Property~2.

It remains to show that $C_i^{1/\beta^2} \subseteq C \subseteq C_i$. Since $M'(x_i) \cap M'(x) \neq \emptyset$, we can apply Lemma~\ref{lem:cap-cap} to caps $A_i$ and $C^{1/\beta}$ (for $\lambda = 1$) to obtain $A_i \subseteq (C^{1/\beta})^{\beta}$. Applying Lemma~\ref{lem:cap-cap} again to caps $C^{1/\beta}$ and $A_i$ (for $\lambda = \beta$), we obtain $(C^{1/\beta})^{\beta} \subseteq A_i^{\beta^2}$. Thus $A_i \subseteq C \subseteq A_i^{\beta^2}$. Recalling that $C_i = A_i^{\beta^2}$, we obtain $C_i^{1/\beta^2} \subseteq C \subseteq C_i$, as desired.
\end{proof}

\section{Polytope Approximation} \label{s:approx}
In this section, we will show how to obtain an $\eps$-approximating convex polytope $P$ of low combinatorial complexity. Let $K$ be a convex body in canonical form. Our strategy is as follows. First, we build a set $\RR$ of disjoint centrally symmetric convex bodies lying within $K$ and close to its boundary. These bodies will possess certain key properties to be specified later. For each $R \in \RR$, we select a point arbitrarily from this body, and let $S$ denote this set of points. The approximation $P$ is defined as the convex hull of $S$. In Lemma~\ref{lem:apx}, we will prove that $P$ is an $\eps$-approximation of $K$ and, in Lemma~\ref{lem:fewfaces}, we will apply a deterministic variant of the witness-collector approach~\cite{DGG13} to show that $P$ has low combinatorial complexity.

Before delving into the details, we provide a high-level overview of the witness-collector method, adapted to our context.
Let $\HH$ denote the set of all halfspaces in $\RE^d$. We define a set $\WW$ of regions called \emph{witnesses} and a set $\CC$ of regions called \emph{collectors}, which satisfy the following properties:
\begin{enumerate}
\item[(1)] Each witness of $\WW$ contains a point of $S$ in its interior.
\item[(2)] Any halfspace $H \in \HH$ either contains a witness $W \in \WW$ or $H \cap S$ is contained in a collector $C \in \CC$.
\item[(3)] Each collector $C \in \CC$ contains a constant number of points of $S$.
\end{enumerate}

The key idea of the witness-collector method is encapsulated in the following lemma.

\begin{lemma} \label{lem:witness-collector}
Given a set of witnesses and collectors satisfying the above properties, the combinatorial complexity of the convex hull $P$ of $S$ is $O(|\CC|)$.
\end{lemma}

\begin{proof}
We map each face $f$ of $P$ to any maximal subset $S_f \subseteq S$ of affinely independent points on $f$. Note that this is a one-to-one mapping and $|S_f| \le d$. In order to bound the combinatorial complexity of $P$ it suffices to bound the number of such subsets $S_f$. 

For a given face $f$, let $H$ be any halfspace such that $H \cap P = f$. Clearly $H$ does not contain any witness since otherwise, by Property~1, it would contain a point of $S$ in its interior. By Property~2, $H \cap S$ is contained in some collector $C \in \CC$. Thus $S_f \subseteq C$. Since $|S_f| \le d$, it follows that the number of such subsets $S_f$ that are contained in any collector $C$ is at most
\[
  \sum_{1 \le j \le d} {\binom{|C \cap S|}{j}} 
	~ = ~ O(|C \cap S|^d) 
	~ = ~ O(1),
\]
where in the last step we have used the fact that $|C \cap S|  = O(1)$ (Property~3). Summing over all the collectors, it follows that the total number of sets $S_f$, and hence the combinatorial complexity of $P$, is $O(|\CC|)$.
\end{proof}

A natural choice for the witnesses and collectors would be the convex bodies $R_i$ and the caps $C_i$, respectively, from Lemma~\ref{lem:ecc}. Unfortunately, these bodies do not work for our purposes. The main difficulty is that Property~3 could fail, since a cap $C_i$ could intersect a non-constant number of bodies of $\RR$, and hence contain a non-constant number of points of $S$. 
(To see this, suppose that $K$ is a cylinder in $3$-dimensional space. A cap of width $\Theta(\eps)$ that is parallel to the circular flat face of $K$ intersects $\Omega(1/\sqrt{\eps})$ bodies, which will be distributed around the circular boundary of this face.) In this section, we show that it is possible to construct a set of witnesses and collectors that satisfy all the requirements by scaling and translating the convex bodies from Lemma~\ref{lem:ecc} into a stratified placement according to their volumes. The properties we obtain are specified below in Lemma~\ref{lem:layers}.

We begin with some easy geometric facts about a convex body $K$ in canonical form.  For any point $x \in K$, define $\delta(x)$ to be the minimum distance from $x$ to any point on $\partial K$. Further, define the \emph{ray-distance} of a point $x$ to the boundary as follows. Consider the ray emanating from $O$ and passing through $x$. Let $p$ denote the intersection of this ray with $\partial K$. We define $\ray(x) = \|xp\|$. Clearly $\ray(x) \ge \delta(x)$. Lemma~\ref{lem:raydist-delta} shows that these two quantities are the same to within a constant factor.

\begin{lemma} \label{lem:raydist-delta}
Let $K \subset \RE^d$ be a convex body in canonical form. For any point $x \in K$, $\ray(x) \le d \cdot \delta(x)$.
\end{lemma}

\begin{proof}
Let $p$ denote the intersection with $\partial K$ of the ray emanating from $O$ and passing through $x$ (see Figure~\ref{f:raydist-delta}(a)). Let $K'$ denote the convex hull of the point $p$ and the ball $B^{1/2d}(O)$. By convexity, $K'$ contains the segment $Op$ and $K' \subseteq K$. It follows that the distance between $x$ and $\partial K'$ is a lower bound on $\delta(x)$. 

\begin{figure}[htbp]
  \centerline{\includegraphics[scale=.75]{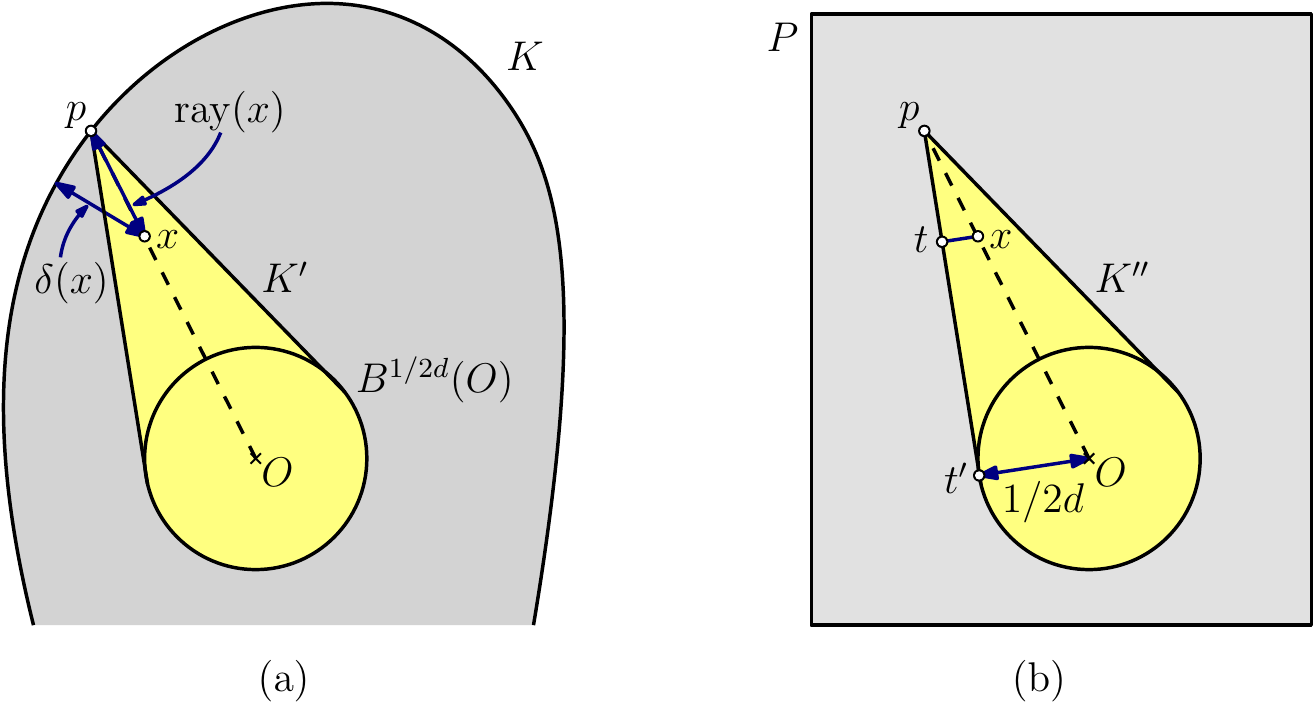}}
  \caption{\label{f:raydist-delta}Illustrating Lemma~\ref{lem:raydist-delta}.}
\end{figure}

To compute the distance between $x$ and $\partial K'$, consider any 2-flat $P$ containing the line $Op$ and let $K'' = K' \cap P$ (see Figure~\ref{f:raydist-delta}(b)). By symmetry, the distance between $x$ and $\partial K'$ is the same as the distance between $x$ and $\partial K''$. Note that $\partial K''$ consists of a portion of a circle of radius $1/(2d)$ centered at $O$, and the two tangents to this circle from point $p$. It is straightforward to see that the points of $\partial K''$ that are closest to $x$ lie on the two tangent lines (one on each tangent). Let $t'$ denote the point where one of these tangents touches the circle, and let $t$ denote the point on segment $pt'$ that is closest to $x$. Since triangles $\triangle Ot'p$ and $\triangle xtp$ are similar, we have $\|xp\| / \|xt\| = \|Op\| / \|Ot'\|$. Since $\|Op\| \le 1/2$ and $\|Ot'\| \ge 1/(2d)$, we have $\|xp\| / \|xt\| \le d$. That is, $\ray(x) = \|xp\| \le d \|xt\| \le d \cdot \delta(x)$, as desired.
\end{proof}

The following technical lemma gives upper and lower bounds on the volume of a cap of width $\alpha$.

\begin{lemma} \label{lem:width-vol}
Let $K \subset \RE^d$ be a convex body in canonical form and let $\alpha < 1$ be a positive real. Then the volume of any cap $C$ of width $\alpha$ is $O(\alpha)$ and $\Omega(\alpha^d)$.
\end{lemma}

\begin{proof}
Let $h_1$ be the hyperplane passing through the base of $C$ and let $h_2$ be the parallel hyperplane passing through the apex $x$ of $C$. Since $C$ is contained in the intersection of ball $B^{1/2}(O)$ with the slab bounded by $h_1$ and $h_2$, it follows that $\vol(C) = O(\alpha)$.

To prove the lower bound, let $y$ denote the point where the ray $Ox$ intersects the base of the cap. We have $\ray(y)  = \|xy\|  \ge \alpha$. By Lemma~\ref{lem:raydist-delta}, we have $\delta(y) \ge \ray(y) / d$. It follows that $\delta(y) \ge \alpha / d$. Note that the ball of radius $\delta(y)$ centered at $y$ is contained within $K$ and half this ball lies within the cap $C$. Therefore, $\vol(C) = \Omega(\alpha^d)$.
\end{proof}

The following lemma states that containment of caps is preserved if the halfspaces defining both caps are consistently scaled about a point that is common to both caps.

\begin{lemma} \label{lem:sandwich}
Let $K$ be a convex body and let $\lambda \ge 1$. Let $C_1$ and $C_2$ be two caps of $K$ such that $C_1 \subseteq C_2$.  Let $H_1$ and $H_2$ be the defining halfspaces of $C_1$ and $C_2$, respectively. Let $H'_1$ and $H'_2$ be the halfspaces obtained by scaling $H_1$ and $H_2$, respectively, by a factor of $\lambda$ about $p$, where $p$ is any point in $K \cap C_1$. Let $C'_1$ and $C'_2$ be the caps $K \cap H'_1$ and $K \cap H'_2$, respectively. Then $C'_1 \subseteq C'_2$.
\end{lemma}

\begin{proof}
Given $\lambda$ and $p$, consider the affine transformation $f(q) = \lambda(q - p) + p$, which scales space by a factor of $\lambda$ about $p$. Thus, $H'_1 = f(H_1)$ and $H'_2 = f(H_2)$, Since $p \in K$ and $\lambda \ge 1$, it follows directly from convexity that $K \subseteq f(K)$. Given any halfspace $H$ such that $p \in K \cap H$, it follows that $K \cap f(H) = K \cap f(K \cap H)$. Since, $C_1 \subseteq C_2$, we have $f(K \cap H_1) \subseteq f(K \cap H_2)$, and thus,
\[
	C'_1 
		~ = ~ K \cap f(H_1) 
		~ = ~ K \cap f(K \cap H_1) 
		~ \subseteq ~ K \cap f(K \cap H_2) 
		~ = ~ K \cap f(H_2) 
		~ = ~ C'_2, 
\]
as desired.
\end{proof}

Our choice of witnesses and collectors will be based on the following lemma. Specifically, the convex bodies $R_1, \ldots, R_k$, will play the role of the witnesses and the regions $C_1, \ldots, C_k$, will play the role of the collectors. The lemma strengthens Lemma~\ref{lem:ecc}, achieving the critical property that any collector $C_i$ intersects only a constant number of convex bodies of $\RR$. As each witness set $R_i$ will contain one point, this ensures that a collector contains only a constant number of input points (Property~3 of the witness-collector system). This strengthening is achieved at the expense of only an extra polylogarithmic factor in the number of collectors needed, compared with Lemma~\ref{lem:ecc}. Also, the collectors are no longer simple caps, but have a more complex shape as described in the proof (this, however, has no adverse effect in our application). 

\begin{lemma} \label{lem:layers}
Let $\eps > 0$ be a sufficiently small parameter, and $\widehat{\eps} = \eps / \log(1/\eps)$. Let $K \subset \RE^d$ be a convex body in canonical form.  There exists a collection $\RR$ of $k = O(1/\widehat{\eps}^{\kern+1pt (d-1)/2})$ disjoint centrally symmetric convex bodies $R_1, \ldots, R_k$ and associated regions $C_1,\ldots, C_k$ such that the following hold:
\begin{enumerate}
\item Let $C$ be any cap of width $\eps$. Then there is an $i$ such that $R_i \subseteq C$.

\item Let $C$ be any cap. Then there is an $i$ such that either (i) $R_i \subseteq C$ or (ii) $C \subseteq C_i$.

\item For each $i$, the region $C_i$ intersects at most a constant number of bodies of $\RR$.
\end{enumerate}
\end{lemma}

As mentioned earlier, our proof of this lemma is based on a stratified placement of the convex bodies from Lemma~\ref{lem:ecc}, which are distributed among $O(\log(1/\eps))$ layers that lie close to the boundary of $K$. Let $\alpha = c_1 \, \eps / \log(1/\eps)$, where $c_1$ is a suitable constant to be specified later. We begin by applying Lemma~\ref{lem:ecc} to $K$ using $\eps = \alpha$. This yields a collection $\RR'$ of $k = O(1/\alpha^{(d-1)/2})$ disjoint centrally symmetric convex bodies $\{R'_1, \ldots, R'_k\}$ and associated caps $\CC' = \{C'_1, \ldots, C'_k\}$. Our definition of the convex bodies $R_i$ and regions $C_i$ required in Lemma~\ref{lem:layers} will be based on $R'_i$ and $C'_i$, respectively. In particular, the convex body $R_i$ will be obtained by translating a scaled copy of $R'_i$ into an appropriate layer, based on the volume of $R'_i$. 

Before describing the construction of the layers, it will be convenient to group the bodies in $\RR'$ based on their volumes. We claim that the volume of any convex body $R'_i$ lies between $c_2 \alpha^d$ and $c_3 \alpha$ for suitable constants $c_2$ and $c_3$. By Property~1 of Lemma~\ref{lem:ecc}, $R'_i \subseteq C'_i \subseteq (R'_i)^{\lambda}$ and $C'_i$ has width $\beta\alpha$, for constants $\beta$ and $\lambda$ depending only on $d$.  By Lemma~\ref{lem:width-vol}, the volume of $C'_i$ is $O(\alpha)$ and $\Omega(\alpha^d)$. Since $\vol(R'_i) = \Theta(\vol(C'_i))$, the desired claim follows.

We partition the set $\RR'$ of convex bodies into $t$ groups, where each group contains bodies whose volumes differ by a factor of at most 2. More precisely, for $0 \le j \le t-1$, group $j$ consists of bodies in $\RR'$ whose volume lies between $c_3 \alpha / 2^j$ and $c_3 \alpha / 2^{j+1}$.  The lower and upper bound on the volume of bodies in $\RR'$ implies that the number of groups $t$ can be expressed as $\floor{c_4 \log(1/\alpha)}$ for a suitable constant $c_4$ (depending on $c_2$ and $c_3$). 

Next we describe how the layers are constructed. We will construct $t$ layers corresponding to the $t$ groups of $\RR'$.  Let $\gamma = 1 - 4d \beta \alpha$. For $0 \le j \le t$, let $T_j$ denote the linear transformation that represents a uniform scaling by a factor of $\gamma^j$ about the origin, and let $K_j = T_j(K)$ (see Figure~\ref{f:layers}(a)). Note that $K_0 = T_0(K) = K$. For $0 \le j \le t-1$, define layer $j$, denoted $L_j$, to be the difference $K_{j} \setminus K_{j+1}$. Whenever we refer parallel supporting hyperplanes for two bodies $K_i$ and $K_j$, we assume that both hyperplanes lie on the same side of the origin.

\begin{figure}[tbp]
  \centerline{\includegraphics[scale=.75]{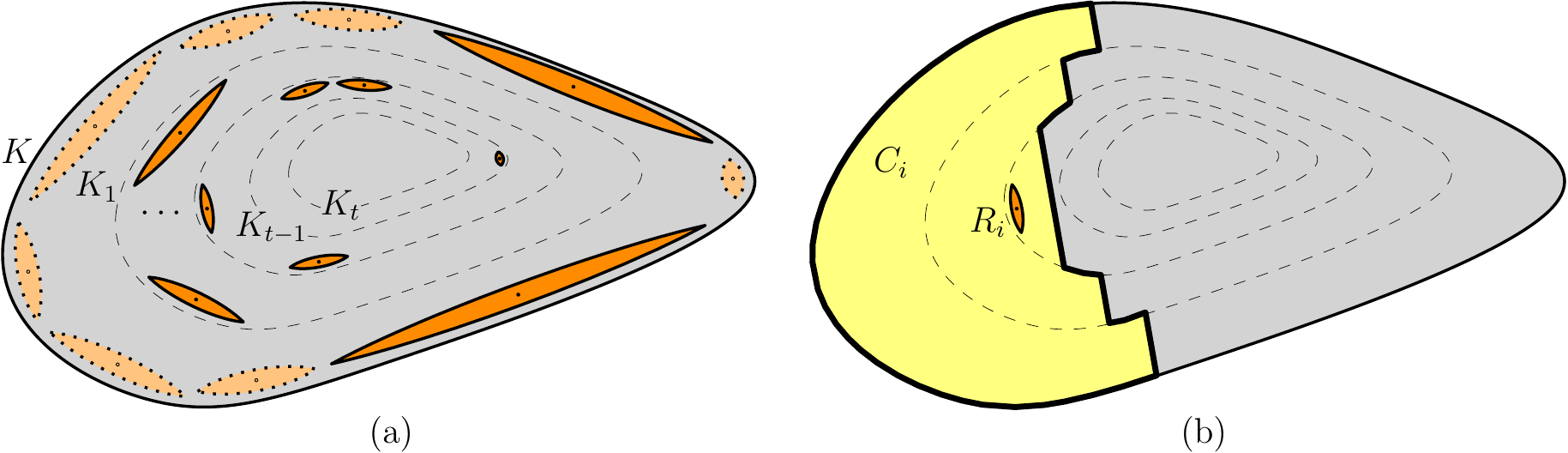}}
  \caption{\label{f:layers}(a) Stratified placement of the bodies $R_i$ and (b) the region $C_i$ corresponding to a body $R_i$. (Figure not to scale.)}
\end{figure}

The following lemma describes some straightforward properties of these layers and the scaling transformations. In particular, the lemma shows that the $t$ layers lie close to the boundary of $K$ (within distance $\eps$) and each layer has a ``thickness'' of $\Theta(\alpha)$.

\begin{lemma} \label{lem:scale}
Let $\eps > 0$ be a sufficiently small parameter. For sufficiently small constant $c_1$ in the definition of $\alpha$ (depending on $c_4$, $\beta$, and $d$), the layered decomposition and the scaling transformations described above satisfy the following properties:
 \begin{enumerate}[label=(\alph*)]
 \item\label{a} For $0 \le j \le t-1$, the distance between parallel supporting hyperplanes of $K_j$ and $K_{j+1}$ is at most $2d \beta \alpha$.
 
 \item\label{b}  For $0 \le j \le t-1$, the distance between parallel supporting hyperplanes of $K_j$ and $K_{j+1}$ is at least $\beta \alpha$.

 \item\label{c} The distance between parallel supporting hyperplanes of $K$ and $K_t$ is at most $\eps$.
  
 \item\label{d} For $0 \le j \le t$, the scaling factor for $T_j$ is at least 1/2 and at most 1.
  
 \item\label{e} For $0 \le j \le t$, $T_j$ preserves volumes up to a constant factor.
 
 \item\label{f} For $0 \le j \le t$, and any point $p \in K$, the distance between $p$ and $T_j(p)$ is at most $2 j d \beta \alpha$.
 \end{enumerate}

\end{lemma}

\begin{proof}
To prove \ref{a}, let $h_1,h_2$ denote parallel supporting hyperplanes of $K_j,K_{j+1}$, respectively. Since $K$ is in canonical form, and the scaling factor of the transformation $T_j$ is at most 1, it follows that $h_1$ is at distance at most $1/2$ from the origin.  Since $h_2$ is the hyperplane obtained by scaling $h_1$ by a factor of $1-4d\beta\alpha$ about the origin, it follows that the distance between $h_1$ and $h_2$ is at most $2 d \beta \alpha$.

To prove \ref{c}, let $h_1,h_2$ denote parallel supporting hyperplanes of $K, K_t$, respectively. The upper bound of \ref{a} implies that the distance between $h_1$ and $h_2$ is at most $2 t d  \beta \alpha$. Recall that $t \le c_4 \log(1/\alpha)$ and $\alpha = c_1 \, \eps / \log(1/\eps)$. By choosing a sufficiently small constant $c_1$ in the definition of $\alpha$ (depending on $d, c_4$ and $\beta$), we can ensure that the distance between $h_1$ and $h_2$ is at most $2 t d \beta \alpha \le \eps$. 

In the rest of this proof, we will assume that $c_1$ in the definition of $\alpha$ is sufficiently small, so \ref{c} holds. To prove \ref{d}, note that we only need to show the lower bound on the scaling factor of $T_j$, since the upper bound is obvious. Again, let $h_1,h_2$ denote parallel supporting hyperplanes of $K, K_t$, respectively. Since $K$ is in canonical position, $h_1$ is at distance at least $1/(2d)$ from the origin. Recall that $T_t$ maps $h_1$ to $h_2$ and, as shown above, the distance between $h_1$ and $h_2$ is at most $\eps$. It follows that the scaling factor of $T_t$ is at least $1 - \eps / (1/2d) = 1 - 2d\eps$. By choosing $\eps$ sufficiently small, we can ensure that the scaling factor of $T_t$ is at least $1/2$. Clearly, this lower bound on the scaling factor also applies to any transformation $T_j$, $0 \le j \le t$. This proves \ref{d}. Note that \ref{e} is an immediate consequence.

To prove \ref{b}, let $h_1,h_2$ denote parallel supporting hyperplanes of $K_j,K_{j+1}$, respectively. Let $h'_1, h'_2$, denote the corresponding supporting hyperplanes of $K,K_1$, respectively. That is, $h_1 = T_j(h'_1)$ and $h_2 = T_j(h'_2)$. Since $K$ is in canonical form, $h'_1$ is at distance at least $1/(2d)$ from the origin. As $h'_2$ is obtained by scaling $h'_1$ by a factor of $1-4d\beta\alpha$ about the origin, it follows that the distance between $h'_1$ and $h'_2$ is at least $2 \beta \alpha$. Since $h_1 = T_j(h'_1)$ and $h_2 = T_j(h'_2)$ and, by \ref{d}, the scaling factor of $T_j$ is at least $1/2$, \ref{b} follows.

Finally, to prove \ref{f}, note that the distance of $p$ from the origin is at most $1/2$. It follows that applying $T_1$ to $p$ moves it closer to the origin by a distance of at most $2 d\beta \alpha$. Since $T_j = (T_1)^j$, \ref{f} follows.
\end{proof}

We are now ready to define the regions $R_i$ and $C_i$ required in Lemma~\ref{lem:layers}. Suppose that $R'_i$ is in group $j$ and let $C'_i = K \cap H'_i$, where $H'_i$ is a halfspace. We define $R_i = T_j(R'_i)$. In order to define $C_i$, we first define caps $C_{i,r}$ of $K_r$ as $C_{i,r} = K_r \cap T_j(H'_i)$ for $0 \le r \le j$. We then define 
\[
  C_i ~ = ~ \bigcup_{r = 0}^{j} C_{i,r}^{\sigma} \cap L_r,
\]
where $\sigma = 4 d \beta^2$. (See Figure~\ref{f:layers}(b).) 

In Lemma~\ref{lem:prop0}, we show that the regions $R_i$ are contained in layer $j$ if $R'_i$ is in group $j$. In Lemma~\ref{lem:prop12}, we establish Properties 1 and 2 of Lemma~\ref{lem:layers}. Finally, in Lemma~\ref{lem:prop3}, we establish Property~3 of Lemma~\ref{lem:layers}.

\begin{lemma} \label{lem:prop0}
Let $R_i \in \RR$. If $R'_i$ is in group $j$, then $C_{i,j} = T_j(C'_i)$ and $R_i \subseteq C_{i,j} \subseteq L_j$.
\end{lemma}

\begin{proof}
Let $H'_i$ denote the halfspace as defined above, that is, $C'_i = K \cap H'_i$. By definition, $C_{i,j} = K_j \cap T_j(H'_i) = T_j(K \cap H'_i) = T_j(C'_i)$. By Property~1 of Lemma~\ref{lem:ecc}, $R'_i \subseteq C'_i$ and $C'_i$ is a cap of $K$ of width $\beta \alpha$. By Lemma~\ref{lem:scale}\ref{b}, the distance between any parallel supporting hyperplanes of $K$ and $K_1$, respectively, is at least $\beta \alpha$. It follows that $R'_i \subseteq C'_i \subseteq L_0 = K \setminus K_1$. Applying the transformation $T_j$ to all these sets yields $R_i \subseteq C_{i,j} \subseteq L_j = K_j \setminus K_{j+1}$. This completes the proof.
\end{proof}

\begin{figure}[tbp]
  \centerline{\includegraphics[scale=.75]{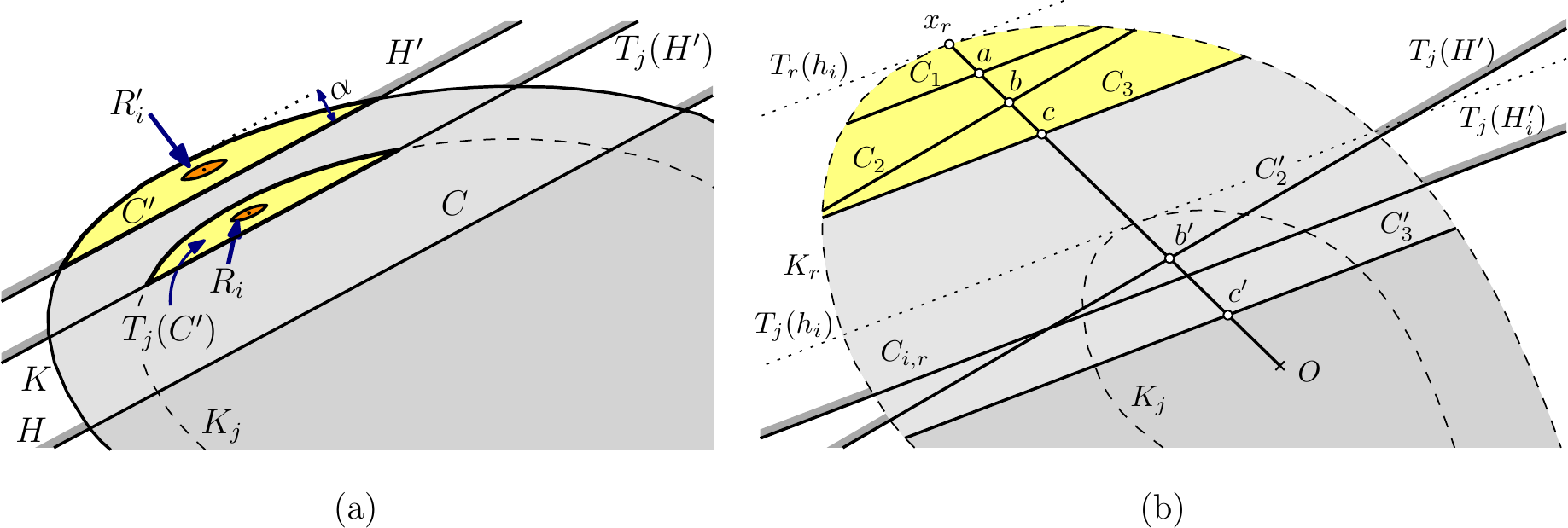}}
  \caption{\label{f:prop12}Proof of Lemma~\ref{lem:prop12} (a) Case~1 and (b) Case~2. (Figure not to scale.)}
\end{figure}

\begin{lemma} \label{lem:prop12}
 Let $C$ be any cap of $K$. Then there is an $i$ such that either (i) $R_i \subseteq C$ or (ii) $C \subseteq C_i$. Furthermore, if the width of $C$ is $\eps$, then (i) holds.
\end{lemma}

\begin{proof}
Let $C' \subseteq C$ be the cap of width $\alpha$, whose base is parallel to the base of $C$. Let $H$ and $H'$ denote the defining halfspaces of $C$ and $C'$, respectively. By Property~2 of Lemma~\ref{lem:ecc}, there is an $i$ such that $R'_i \subseteq C'$. Suppose that $R'_i$ is in group $j$.  We consider two cases, depending on whether $T_j(H') \subseteq H$ or $H \subset T_j(H')$. To complete the proof of the lemma, we will show that in the former case, $R_i \subseteq C$ and, in the latter case, $C \subseteq C_i$. Additionally, we will show that if $C$ has width $\eps$, then the former case holds (implying that $R_i \subseteq C$).

\medskip\noindent\textbf{Case~1: $T_j(H') \subseteq H$.}
Arguing as in the proof of Lemma~\ref{lem:prop0} (but with $C'$ in place of $C'_i$), we have $R_i \subseteq T_j(C') = K_j \cap T_j(H') \subseteq L_j$ (see Figure~\ref{f:prop12}(a)). Observe that $K_j \cap T_j(H')  \subseteq K \cap H = C$. Therefore $R_i \subseteq C$. 

Also, by Lemma~\ref{lem:scale}\ref{c}, the distance between any parallel supporting hyperplanes of $K$ and $K_t$ is at most $\eps$. Since $K_j \cap T_j(H') \subseteq L_j$, it follows that the width of cap $K \cap T_j(H')$ is at most $\eps$. Therefore, if $C$ has width $\eps$, then $T_j(H') \subseteq H$ and Case~1 holds.

\medskip\noindent\textbf{Case~2: $H \subset T_j(H')$.} Recall that we need to show that $C \subseteq C_i$. Clearly, it suffices to show that $K \cap T_j(H') \subseteq C_i$ since $C = K \cap H \subset K \cap T_j(H')$. In turn, the definition of $C_i$ implies that it suffices to show that for $0 \le r \le j$, $T_j(H') \cap K_r ~\subseteq~ C_{i,r}^{\sigma}$. 

By Property~2 of Lemma~\ref{lem:ecc}, there is an $i$ such that $(C'_i)^{\phi}  \subseteq C' \subseteq C'_i$, where $\phi = 1/\beta^2$. By Property~1 of Lemma~\ref{lem:ecc}, the widths of the caps $(C'_i)^{\phi} $ and $C'_i$ are $\alpha / \beta$ and $\beta \alpha$, respectively. Recall that $H'_i$ denotes the defining halfspace for the cap $C'_i$.  Also, let $x$ denote the apex of $C'_i$, and let $h_i$ denote the supporting hyperplane to $K$ passing through $x$ and parallel to $C'_i$'s base.

Let $C_1, C_2$, and $C_3$ denote the caps of $K_r$ obtained by applying the transformation $T_r$ to the caps $(C'_i)^{\phi}$, $C'$, and $C'_i$, respectively (see Figure~\ref{f:prop12}(b)). We have $C_1 \subseteq C_2 \subseteq C_3$. Let $a, b$ and $c$ denote the point of intersection of the bases of the caps $C_1, C_2$ and  $C_3$, respectively, with the line segment $Ox$. Let $b'$ denote the point of intersection of the base of the cap $K \cap T_j(H')$ with the segment $Ox$. Let $x_r$ denote the point $T_r(x)$. Consider scaling caps $C_2$ and $C_3$ as described in Lemma~\ref{lem:sandwich}, about the point $x_r$ with scaling factor $\rho = \|b' x_r\| / \|b x_r\|$. Let $C'_2$ and $C'_3$ denote the caps of $K_r$ obtained from $C_2$ and $C_3$, respectively, through this transformation. By Lemma~\ref{lem:sandwich}, $C'_2 \subseteq C'_3$. Our choice of the scaling factor implies that $C'_2$ is the cap $T_j(H') \cap K_r$. We claim that $C'_3 \subseteq C_{i,r}^{\sigma}$. Note that this claim would imply that $T_j(H') \cap K_r \subseteq C_{i,r}^{\sigma}$, and complete the proof.

To prove the above claim, we first show that $\rho = O(j-r+1)$. Observe that $\rho = (\|b' b\| + \|b x_r\|) / \|b x_r\| = \|b' b\| / \|b x_r \| + 1$. We have 
\[
	\|b x_r\| 
		~ \ge ~ \|a x_r\| 
		~ \ge ~ \width(C_1) 
		~ \ge ~ \frac{\width((C'_i)^{\phi})}{2} 
		~ \ge ~ \frac{\alpha}{2 \beta},
\]
where in the third inequality, we have used Lemma~\ref{lem:scale}\ref{d} and the fact that $C_1 = T_r((C'_i)^{\phi})$. Also, since $T_{j-r}(b) = b'$, it follows from Lemma~\ref{lem:scale}\ref{f} that $\|b' b\|$ is at most $2 (j-r) d \beta \alpha$. Substituting the derived bounds on $\|b' b\|$ and $\|b x_r\|$, we obtain $\rho \le 4 d \beta^2 (j-r) + 1$. 

Recall that $C'_3$ and $C_{i,r}$ are caps of $K_r$ defined by parallel halfspaces. To prove that $C'_3 \subseteq C_{i,r}^{\sigma}$, it therefore suffices to show that $\width(C'_3) /  \width(C_{i,r}) \le \sigma$. We have
\[
 	\width(C'_3) 
		~  =  ~ \rho \cdot \width(C_3) 
		~ \le ~ \rho \cdot \width(C'_i)
		~  =  ~ \rho \kern+1pt \beta \kern+1pt \alpha,
\]
where in the second step, we have used Lemma~\ref{lem:scale}\ref{d} and the fact that $C_3 = T_r(C'_i)$. Also, it is easy to see that the width of $C_{i,r}$ is the sum of the width of the cap $T_j(C'_i)$ and the distance between the hyperplanes $T_r(h_i)$ and $T_j(h_i)$. Since $\width(C'_i) = \beta \alpha$, by Lemma~\ref{lem:scale}\ref{d}, the width of the cap $T_j(C'_i)$ is at least $\beta \alpha / 2$. Also, by Lemma~\ref{lem:scale}\ref{b}, the distance between the hyperplanes $T_r(h_i)$ and $T_j(h_i)$ is at least $(j-r) \beta \alpha$. It follows that the width of $C_{i,r}$ is at least $\beta \alpha / 2 + (j-r) \beta \alpha = (j-r+1/2) \beta \alpha$. Thus, 
\[
	\frac{\width(C'_3)}{\width(C_{i,r})} 
		~ \le ~ \frac{\rho \kern+1pt \beta \kern+1pt \alpha}{(j-r+1/2) \beta \kern+1pt \alpha} 
		~  =  ~ \frac{\rho}{j-r+1/2} \le \frac{4 \kern+1pt d \kern+1pt \beta^2 (j-r)+1}{j-r+1/2}
		~ \le ~ 4 \kern+1pt d \kern+1pt \beta^2 
		~  =  ~ \sigma,
\]
as desired.
\end{proof}

\begin{lemma} \label{lem:prop3}
For each $i$, the region $C_i$ intersects $O(1)$ bodies of $\RR$.
\end{lemma}

\begin{figure}[tbp]
  \centerline{\includegraphics[scale=.75]{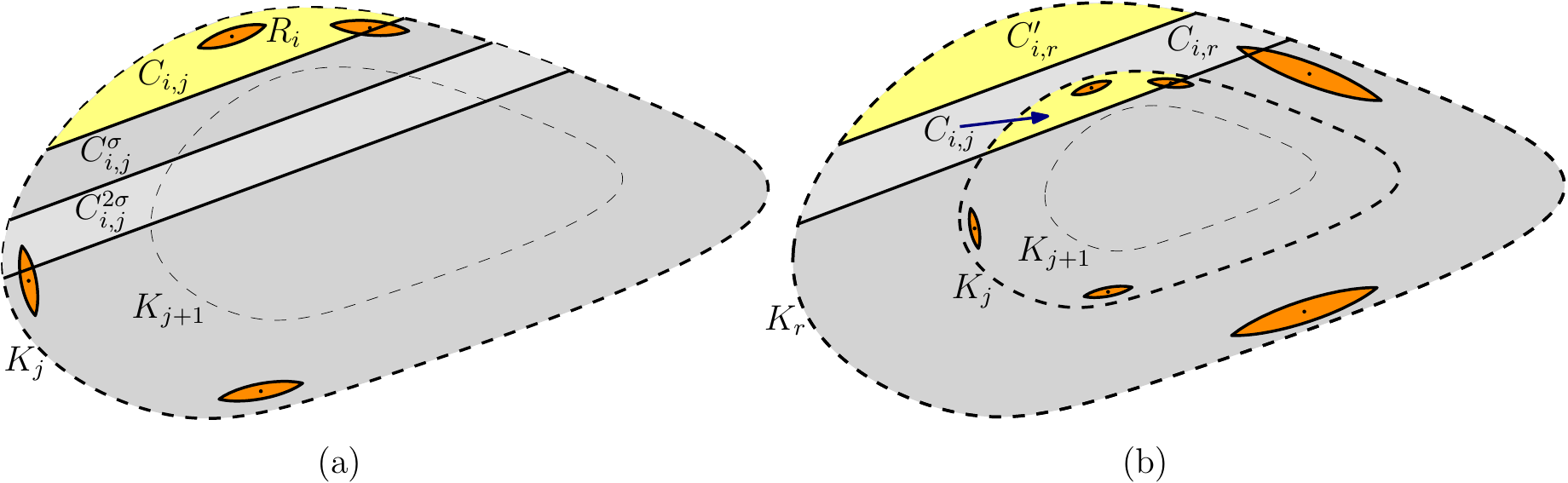}}
  \caption{\label{f:prop3}Proof of Lemma~\ref{lem:prop3}. (Figure not to scale.)}
\end{figure}

\begin{proof}
Suppose that $R'_i$ is in group $j$. Recall that $R_i = T_j(R'_i)$, $C'_i = K \cap H'_i$ and $C_i = \bigcup_{r = 0}^{j} (C_{i,r}^{\sigma} \cap L_r)$. We begin by bounding the number of bodies of $\RR$ that overlap $C_{i,j}^{\sigma} \cap L_j$. (See Figure~\ref{f:prop3}(a).) By Lemma~\ref{lem:prop0}, $C_{i,j} = T_j(C'_i)$ and $R_i \subseteq C_{i,j} \subseteq L_j$.  By Property~1 of Lemma~\ref{lem:ecc}, we have $C'_i \subseteq (R'_i)^{\lambda}$, which implies that $\vol(R'_i) = \Omega(\vol(C'_i))$. Recall that all the bodies of $\RR'$ in group $j$ have the same volumes to within a factor of 2, and so they all have volumes $\Omega(\vol(C'_i))$. By Lemma~\ref{lem:scale}\ref{e}, the scaling transformations used in our construction preserve volumes to within a constant factor. Also, recall that the bodies of $\RR$ in layer $j$ are scaled copies of the bodies of $\RR'$ in group $j$. It follows that the bodies of $\RR$ in layer $j$ all have volumes $\Omega(\vol(C_{i,j}))$.

Next, we assert that any body of $\RR$ that overlaps $C_{i,j}^{\sigma} \cap L_j$ is contained within the cap $C_{i,j}^{2 \sigma}$. To prove this, recall from the proof of Lemma~\ref{lem:ecc} that the bodies of $\RR'$ are  $(1/5)$-scaled disjoint Macbeath regions with respect to $K$. It follows that the bodies of $\RR$ in layer $j$ are $(1/5)$-scaled disjoint Macbeath regions with respect to $K_j$. By Lemma~\ref{lem:cap-mac}, it now follows that any body of $\RR$ that overlaps $C_{i,j}^{\sigma} \cap L_j$ is contained within the cap $C_{i,j}^{2 \sigma}$. Since $\vol(C_{i,j}^{2 \sigma}) = O(\vol(C_{i,j}))$, and all  bodies of $\RR$ in layer $j$ have volumes $\Omega(\vol(C_{i,j}))$, it follows by a simple packing argument that the number of bodies of $\RR$ that overlap $C_{i,j}^{\sigma} \cap L_j$ is $O(1)$.

Next we bound the number of bodies of $\RR$ that overlap $C_{i,r}^{\sigma} \cap L_r$, where $0 \le r < j$. (See Figure~\ref{f:prop3}(b).) Recall that $C_{i,r} = K_r \cap T_j(H'_i)$. Roughly speaking, we will show that the volume of $C_{i,r}$ exceeds the volume of $C_{i,j}$ by a factor that is at most polynomial in $j-r$, while the volume of the bodies in layer $r$ exceeds the volume of the bodies in layer $j$ by a factor that is exponential in $j-r$. This will allow us to show that the number of bodies of $\RR$ that overlap $C_i$ is bounded by a constant. We now present the details.

Define $C'_{i,r} = T_r(C'_i)$. Recall that $C_{i,j} = T_j(C'_i)$. By Lemma~\ref{lem:scale}\ref{e}, $T_j$ and $T_r$ preserve volumes up to constant factors, and so $\vol(C'_{i,r}) = \Theta(\vol(C_{i,j}))$. Since the width of $C'_i$ is $\beta \alpha$, by Lemma~\ref{lem:scale}\ref{d}, it follows that the width of $C'_{i,r}$ is at least $\beta \alpha / 2$. Also, the width of $C_{i,r}$ is upper bounded by the distance between parallel supporting hyperplanes of $K_r$ and $K_{j+1}$ which by Lemma~\ref{lem:scale}\ref{a} is at most $2 d \beta \alpha (j-r+1)$. It follows that the width of $C_{i,r}$ is $O(j-r+1)$ times the width of $C'_{i,r}$. Recalling that, for $\lambda \ge 1$, the volume of a $\lambda$-expansion of a cap is at most $\lambda^d$ times the volume of the cap, it follows that $\vol(C_{i,r}) = O((j-r+1)^d) \cdot \vol(C'_{i,r}) = O((j-r+1)^d) \cdot \vol(C_{i,j})$.

Next, recall that the volume of the bodies of $\RR'$ in group $r$ exceeds the volume of the bodies of $\RR'$ in group $j$ by a factor of $\Omega(2^{j-r+1})$. It follows from Lemma~\ref{lem:scale}\ref{e} and our construction that the volume of the bodies of $\RR$ in layer $r$ exceeds the volume of the bodies of $\RR$ in layer $j$ by a factor of $\Omega(2^{j-r+1})$. For the same reasons as discussed above, any body of $\RR$ that overlaps $C_{i,r}^{\sigma} \cap L_r$ is contained within $C_{i,r}^{2 \sigma}$, and $\vol(C_{i,r}^{2 \sigma}) = O(\vol(C_{i,r}))$. Putting this together with the upper bound on $\vol(C_{i,r})$ shown above, we have $\vol(C_{i,r}^{2 \sigma}) = O((j-r+1)^d) \cdot \vol(C_{i,j})$. By a simple packing argument, it follows that the ratio of the number of bodies of $\RR$ that overlap $C_{i,r}^{\sigma} \cap L_r$ to the number of bodies of $\RR$ that overlap $C_{i,j}^{\sigma} \cap L_j$ is $O((j-r+1)^d / 2^{j-r+1})$. Recall that the number of bodies of $\RR$ that overlap $C_{i,j}^{\sigma} \cap L_j$ is $O(1)$. It follows that the number of bodies of $\RR$ that overlap $C_i =  \bigcup_{r = 0}^{j} (C_{i,r}^{\sigma} \cap L_r)$ is on the order of $\sum_{0 \le r \le j} (j-r+1)^d / 2^{j-r+1} = O(1)$, as desired.
\end{proof}

Let $S$ be a set of points containing one point inside each body of $\RR$ defined in Lemma~\ref{lem:layers} and no other points.

\begin{lemma} \label{lem:apx}
The polytope $P = \conv(S)$ is an $\eps$-approximation of $K$.
\end{lemma}

\begin{proof}
A set of points $S$ \emph{stabs} every cap of width $\eps$ if every such cap contains at least one point of $S$. It is well known that if a set of points $S \subset K$ stabs all caps of width $\eps$ of $K$, then $\conv(S)$ is an $\eps$-approximation of $K$~\cite{BrI76}. Let $C$ be a cap of width $\eps$. By Lemma~\ref{lem:layers}, Property~1, there is a convex body $R_i \subseteq C$. Since $S$ contains a point that is in $R_i$, we have that the cap $C$ is stabbed.
\end{proof}

To bound the combinatorial complexity of $\conv(S)$, and hence conclude the proof of Theorem~\ref{thm:main}, we use the witness-collector approach~\cite{DGG13}.

\begin{lemma} \label{lem:fewfaces}
The number of faces of $P=\conv(S)$ is $O(1/\widehat{\eps}^{\kern+1pt (d-1)/2})$.
\end{lemma}

\begin{proof}
Define the witness set $\WW = R_1, \ldots, R_k$ and the collector set $\CC = C_1, \ldots, C_k$, where the $R_i$'s and $C_i$'s are as defined in Lemma~\ref{lem:layers}. As there is a point of $S$ in each body $R_i$, Property~1 of the witness-collector method is satisfied. To prove Property~2, let $H$ be any halfspace. If $H$ does not intersect $K$, then Property~2 of the witness-collector method holds trivially. Otherwise let $C = K \cap H$. By Property~2 of Lemma~\ref{lem:layers}, there is an $i$ such that either $R_i \subseteq C$ or $C \subseteq C_i$. It follows that $H$ contains witness $R_i$ or $H \cap S$ is contained in collector $C_i$. Thus Property~2 of the witness-collector method is satisfied. Finally, Property~3 of Lemma~\ref{lem:layers} implies Property~3 of the witness-collector method. Thus, we can apply Lemma~\ref{lem:witness-collector} to conclude that the number of faces of $P$ is $O(|\CC|) = O(k)$, which proves the lemma.
\end{proof}

\section{Conclusions and Open Problems} \label{s:conclusion}

We considered the problem of $\eps$-approximating a convex body $K \subset \RE^d$ by a polytope $P$ of small combinatorial complexity. We proved an upper bound of $\tilde{O}(1/\eps^{(d-1)/2})$ to the combinatorial complexity, almost a square-root improvement over the previous bound of $O(1/\eps^{d(d-1)/(d+1)}) \approx O(1/\eps^{d - 2})$. Our bound is optimal up to logarithmic factors. Two natural questions arise. First, can the logarithmic factors be removed or is there a fundamental reason why they appear? Second, our construction is much more complex than the ones of Dudley or Bronshteyn and Ivanov. Can we show that those simpler constructions also attain a low combinatorial complexity or find a counterexample? Furthermore, our bounds are purely existential. While our construction can be turned into an algorithm, there are a number of nontrivial technical issues that would need to be handled in order to obtain an efficient solution.

Our bounds are presented as a function of $\eps$, but a natural question is whether it is possible to obtain bounds that are sensitive to the polytope being approximated. One may consider finding the polytope of minimum combinatorial complexity that approximates a given polytope $K$ as an optimization problem. Approximation algorithms for minimizing the number of vertices of an $\eps$-approximating polytope are well known~\cite{Cla93,MS}, but we know of no similar results for minimizing the combinatorial complexity.

\section*{Acknowledgments} 

We would like to thank the reviewers (of both the conference and journal versions) for their many valuable suggestions. The work of S.\ Arya was supported by the Research Grants Council of Hong Kong, China under project number 610012. The work of D.\ M.\ Mount was supported by NSF grants CCF-1117259 and CCF-1618866. A preliminary version of this paper appeared in the 32nd International Symposium on Computational Geometry, 2016.

\bibliographystyle{plain}
\bibliography{shortcuts,polytope}

\end{document}